\documentclass[11pt]{article}

\newif\iffullver
\fullvertrue 

\newcommand{\full}[1]{\iffullver#1\fi}
\newcommand{\short}[1]{\iffullver\else#1\fi}

\usepackage{typearea}
\paperwidth 8.5in \paperheight 11in \typearea{13}

\usepackage{theorem,latexsym,graphicx}
\usepackage{amsmath,amssymb,enumerate}
\usepackage{xspace}
\usepackage{bm}
\usepackage{ifpdf}
\usepackage{shadow,shadethm,color}
\usepackage{algorithm}
\usepackage{algorithmic}
\usepackage{subfigure}
\usepackage{verbatim}
\usepackage{paralist}
\usepackage{algorithm,algorithmic}

\allowdisplaybreaks

\definecolor{Darkblue}{rgb}{0,0,0.4}
\definecolor{Brown}{cmyk}{0,0.81,1.,0.60}
\definecolor{Purple}{cmyk}{0.45,0.86,0,0}
\newcommand{\mydriver}{hypertex}
\ifpdf
 \renewcommand{\mydriver}{pdftex}
\fi
\usepackage[breaklinks,\mydriver]{hyperref}
\hypersetup{colorlinks=true,
            citebordercolor={.6 .6 .6},linkbordercolor={.6 .6 .6},%
citecolor=Darkblue,urlcolor=black,linkcolor=Darkblue,pagecolor=black}
\newcommand{\lref}[2][]{\hyperref[#2]{#1~\ref*{#2}}}


\newtheorem{theorem}{Theorem}[section]

\newtheorem{lemma}[theorem]{Lemma}
\newshadetheorem{lemmashaded}[theorem]{Lemma}
\newtheorem{cl}{Claim}

\numberwithin{algorithm}{section}

\newcommand{\ignore}[1]{}

\newenvironment{proof}{{\bf Proof:  }}{\hfill\rule{2mm}{2mm}}
\newenvironment{proofof}[1]{{\bf Proof of #1:  }}{\hfill\rule{2mm}{2mm}}

\newcommand{\R}[0]{{\ensuremath{\mathbb{R}}}}

\newcommand{\Z}[0]{{\ensuremath{\mathbb{Z}}}}

\def\cin{h} 

\newcommand{\sse}{\subseteq}

\newcommand{\Ev}{{\mathcal{E}}}

\newcommand{\tsty}{\textstyle}

\newcommand{\OPT}{\ensuremath{{\sf opt}}}

\newcommand{\Opt}{\ensuremath{\mathsf{Opt}\xspace}}

\newcommand{\qedsymb}{\hfill{\rule{2mm}{2mm}}}
\renewenvironment{proof}{\begin{trivlist} \item[\hspace{\labelsep}{\bf
\noindent Proof.\/}] }{\qedsymb\end{trivlist}}%

\newcommand{\initOneLiners}{%
    \setlength{\itemsep}{0pt}
    \setlength{\parsep }{0pt}
    \setlength{\topsep }{0pt}
}
\newenvironment{OneLiners}[1][\ensuremath{\bullet}]
    {\begin{list}
        {#1}
        {\initOneLiners}}
    {\end{list}}

\newcommand{\squishlist}{
 \begin{list}{$\bullet$}
  { \setlength{\itemsep}{0pt}
     \setlength{\parsep}{3pt}
     \setlength{\topsep}{3pt}
     \setlength{\partopsep}{0pt}
     \setlength{\leftmargin}{1.5em}
     \setlength{\labelwidth}{1em}
     \setlength{\labelsep}{0.5em} } }

\newcommand{\squishend}{
  \end{list}  }

\newcommand{\E}{\mathbb{E}}
\newcommand{\x}{\mathbf{x}}

\title{Approximating Sparse Covering Integer Programs Online}
\author{Anupam Gupta\thanks{Computer Science Department, Carnegie Mellon University, Pittsburgh, PA 15213, USA. Supported in part by NSF awards CCF-0964474 and CCF-1016799.}
 \and Viswanath Nagarajan\thanks{IBM T.J. Watson Research Center. }}

\date{}

\begin{document}

\maketitle

\begin{abstract}
  A {\em covering integer program} (CIP) is a mathematical program of
  the form:
  \begin{gather*}
    \min \{ c^\top \x \mid A\x \geq \mathbf{1},\; \mathbf{0} \leq \x
    \leq \mathbf{u},\; \x \in \Z^n\}, \label{eq:5}
  \end{gather*}
  where $A \in R_{\geq 0}^{m \times n}, c,u \in \R_{\geq 0}^n$. In the
  online setting, the constraints (i.e., the rows of the constraint
  matrix $A$) arrive over time, and the algorithm can only increase the
  coordinates of $\x$ to maintain feasibility. As an intermediate step,
  we consider solving the \emph{covering linear program} (CLP) online,
  where the requirement $\x \in \Z^n$ is replaced by $\x \in \R^n$.

  \medskip
  Our main results are (a) an $O(\log k)$-competitive online algorithm
  for solving the CLP, and (b) an $O(\log k \cdot \log
  \ell)$-competitive randomized online algorithm for solving the
  CIP. Here $k\le n$ and $\ell\le m$ respectively denote the maximum number of
  non-zero entries in any row and column of the constraint matrix
  $A$. By a result of Feige and Korman, this is the best possible for
  polynomial-time online algorithms, even in the special case of set
  cover (where $A \in \{0,1\}^{m \times n}$ and $c, u \in \{0,1\}^n$).

  \medskip
  The novel ingredient of our approach is to allow the dual variables to
  increase and decrease throughout the course of the algorithm. We show
  that the previous approaches, which either only raise dual variables,
  or lower duals only within a guess-and-double framework, cannot give a
  performance better than $O(\log n)$, even when each constraint only
  has a single variable (i.e., $k = 1$).
\end{abstract}

\section{Introduction} \label{sec:introduction}

Covering Integer Programs (CIPs) have long been studied, giving a very
general framework which captures a wide variety of natural problems. CIPs are
mathematical programs of the following form:
\begin{alignat}{10}
\min \quad \tsty \sum_{i=1}^n c_i x_i \tag{IP1} \label{eq:1} \\
\mbox{subject to:} \quad \tsty  \sum_{i =1}^n a_{ij} x_i \geq 1 & \quad \quad  \forall j\in[m],\\
0\le x_i \leq u_i & \quad \quad  \forall i\in [n],\\
x\in \mathbb{Z}^n. \label{eq:2}
\end{alignat}
Above, all the entries $a_{ij},\,c_i,$ and $u_i$ are non-negative. The
\emph{constraint matrix} is denoted $A=(a_{ij})_{i\in[n],j\in[m]}$.  We
define $k$ to be the {\em row sparsity} of $A$, i.e., the maximum number
of non-zeroes in any constraint $j\in[m]$.  For each row $j\in [m]$ let
$T_j\sse [n]$ denote its non-zero columns; we say that the variables
indexed by $T_j$ ``appear in'' constraint $j$. Let $\ell$ denote the
{\em column sparsity} of $A$, i.e., the maximum number of constraints
that any variable $i\in[n]$ appears in. Dropping the integrality
constraint~(\ref{eq:2}) gives us a covering linear program (CLP).

In this paper we study the \emph{online} version of these problems, where the constraints $j\in[m]$ arrive over time,
and we are required to maintain a monotone (i.e., non-decreasing) feasible solution $\x$ at each point in time. Our
main results are (a) an $O(\log k)$-competitive algorithm for solving CLPs online, and (b) an $O(\log k\cdot \log
\ell)$-competitive randomized online algorithm for CIPs. In settings where $k\ll n$ or $\ell\ll m$ our results give a
significant improvement over the previous best bounds of $O(\log n)$ for CLPs~\cite{BN-MOR}, and $O(\log n\cdot \log m)$ for CIPs that
can be inferred from rounding these LP solutions. Analyzing performance guarantees for covering/packing integer programs in terms of
row ($k$) and column ($\ell$) sparsity has received much attention in the offline setting,
e.g.~\cite{Srinivasan99,Srinivasan06,KY05,PC11,BKNS10}.  This paper obtains tight bounds in terms of these parameters
for {\em online} covering integer programs.

{\bf Our Techniques.} Our algorithms use online primal-dual framework of Buchbinder and Naor~\cite{BN-mono}. To solve
the covering LP, we give an algorithm that monotonically raises the primal. However, we \emph{both raise and lower the dual
variables} over the course of the algorithm; this is unlike typical applications of the online primal-dual approach,
where both primal and dual variables are only increased (except possibly within a ``guess and double'' framework---see
the discussion in the related work section). This approach of lowering duals is crucial for our bound of $O(\log k)$,
since we show a primal-dual gap of $\Omega(\log n)$ for algorithms that lower duals only within the guess-and-double
framework, even when $k = 1$.

The algorithm for covering IP solves the LP relaxation and then rounds
it. It is well-known that the natural LP relaxation is too weak: so we
extend our online CLP algorithm to also handle Knapsack Cover (KC)
inequalities from~\cite{CFLP}. This step has an $O(\log k)$-competitive
ratio. Then, to obtain an integer solution, we adapt the method of {\em
  randomized rounding with alterations} to the online setting. Direct
randomized rounding as in~\cite{AAABN03} results in a worse $O(\log m)$
overhead, so to get the $O(\log \ell)$ loss we use this different
approach.

{\bf Related Work.} The powerful online primal-dual framework has been used to give algorithms for set
cover~\cite{AAABN03}, graph connectivity and cut problems~\cite{AAABN-talg06},
caching~\cite{Y94,BBN-focs07-paging,BBN-stoc08-caching}, packing/covering IPs~\cite{BN-MOR}, and many more problems.
This framework usually consists of two steps: obtaining a fractional solution (to an LP relaxation) online, and
rounding the fractional solution online to an integral solution. (See the monograph of Buchbinder and
Naor~\cite{BN-mono} for a lucid survey.)

In most applications of this framework, the fractional online algorithm
raises both primal and dual variables monotonically, and the competitive
ratio is given by the primal to dual ratio. For CLPs, Buchbinder and
Naor~\cite{BN-MOR} showed that if we increase dual variables
monotonically, the primal-dual gap can be $\Omega(\log
\frac{a_{max}}{a_{min}})$. In order to obtain an $O(\log n)$-competitive
ratio, they used a {\em guess-and-double}
framework~\cite[Theorem~4.1]{BN-MOR} that changes duals in a partly
non-monotone manner as follows: \full{

\begin{quote}}
\emph{The algorithm proceeds in phases, where each phase $r$ corresponds to the primal value being roughly $2^r$. Within a
phase the primal and dual are raised monotonically. But the algorithm resets duals to zero at the beginning of each
phase---this is the only form of dual reduction.}
\full{\end{quote}}

For the special case of fractional set cover (where $A \in \{0,1\}^{m
  \times n}$), they get an improved $O(\log k)$-competitive ratio using
this guess-and-double framework~\cite[Section~5.1]{BN-MOR}. However, we
show in \lref[Appendix]{sec:monotone-duals} that such dual update
processes do not extend to obtain an $o(\log n)$ ratio for general
CLPs. So our algorithm reduces the dual variables more continuously
throughout the algorithm, giving an $O(\log k)$-competitive ratio for
general CLPs.

\emph{Other online algorithms:} Koufogiannakis and Young~\cite{KY09}
gave a $k$-competitive deterministic online algorithm for CIPs based on
a greedy approach; their result holds for a more general class of
constraints and for submodular objectives. Our $O(\log k \log \ell)$
approximation is incomparable to this result. Feige and
Korman~\cite{Korman05} show that no randomized \emph{polynomial-time}
online algorithm can achieve a competitive ratio better than $O(\log k
\log \ell)$.

\emph{Offline algorithms.} CLPs can be solved optimally offline in polynomial time. For CIPs in the absence of variable
upper bounds, randomized rounding gives an $O(\log m)$-approximation ratio. Srinivasan~\cite{Srinivasan99} gave an
improved algorithm using the FKG inequality (where the approximation ratio depends on the optimal LP value).
Srinivasan~\cite{Srinivasan01} also used the method of alterations in context of CIPs and gave an RNC algorithm
achieving the bounds of~\cite{Srinivasan99}. An $O(\log \ell)$-approximation algorithm for CIPs (no upper bounds) was
obtained in~\cite{Srinivasan06} using the Lov\'asz Local Lemma. Using KC-inequalities and the algorithm
from~\cite{Srinivasan06}, Kolliopoulos and Young~\cite{KY05} gave an $O(\log \ell)$-approximation algorithm for CIPs
with variable upper bounds. Our algorithm matches this $O(\log \ell)$ loss in the online setting. Finally, the
knapsack-cover (KC) inequalities were introduced by Carr et al.~\cite{CFLP} to reduce the integrality gap for CIPs.
These were used in~\cite{KY05,CGK10}, and also in an online context by~\cite{BBN-stoc08-caching} for the generalized
caching problem.

\section{An Algorithm for a Special Class for Covering LPs} \label{sec:clps}

In this section, we consider CLPs without upper bounds on the variables:
\full{\begin{alignat*}{10}
\min \tsty \quad \sum_{i=1}^n c_i x_i \\
\tsty \mbox{subject to:} \quad  \sum_{i =1}^n a_{ij} x_i \geq 1 & \quad \quad  \forall j\in[m],\\
x  \geq \textbf{0} &
\end{alignat*}
}
\short{
\begin{gather}
  \tsty \min \left\{ \sum_{i=1}^n c_i x_i \quad \mid \quad \sum_{i} a_{ij} x_i
    \geq 1 ~~\forall j\in[m], ~~ x \geq \textbf{0} \right\} \label{eq:3}
\end{gather}} and give an $O(\log k)$-competitive deterministic online
algorithm for solving such LPs, where $k$ is an (upper bound) on the
row-sparsity of $A = (a_{ij})$. The dual is the packing linear program:
\full{ \begin{alignat*}{10}
\max \tsty \quad \sum_{j=1}^m y_j \\
\tsty \mbox{subject to:} \quad  \sum_{j=1}^m a_{ij} y_j \leq c_i & \quad \quad  \forall i\in[n],\\
y \geq \textbf{0} &
\end{alignat*}
}
\short{
\begin{gather}
  \tsty \max \left\{ \sum_{j = 1}^m y_j \quad \mid \quad \sum_{j: i \in T_j} a_{ij}
    y_j \leq c_i ~~\forall i \in [n], ~~ y \geq \textbf{0} \right\} \label{eq:4}
\end{gather}
}
We assume that $c_i$'s are strictly positive for all $i$, else we can
drop all constraints containing variable $i$.

{\bf Algorithm~I. } In the online algorithm, we want a solution pair $(x,y)$, where we monotonically increase the value
of $x$, but the dual variables can move up or down as needed. We want a feasible primal, and an approximately feasible
dual. The primal update step is the following:
\begin{quote}
  When constraint $\cin$ (i.e., $\sum_i a_{i\cin} x_i \geq 1$) arrives,
  \begin{OneLiners}
  \item[(a)] define $d_{i\cin} = \frac{c_i}{a_{i\cin}}$ for all $i \in [n]$, and
    $d_{m(\cin)} = \min_i d_{i\cin} = \min_{i \in T_\cin} d_{i\cin}$.
  \item[(b)] while $\sum_i a_{i\cin} x_i < 1$, update the $x$'s by
    \[ x_i^{new} \gets \left( 1+ \frac{d_{m(\cin)}}{d_{i\cin}} \right) x_i^{old}
      ~~~+~~~ \frac{1}{k\cdot a_{i\cin}} \frac{d_{m(\cin)}}{d_{i\cin}}, \qquad \forall i\in T_\cin. \]
      Let $t_\cin$ be the number of times this update step is performed
      for constraint $\cin$.
  \end{OneLiners}
\end{quote}

As stated, the algorithm assumes we know $k$, but this is not required.
We can start with the estimate $k = 2$ and increase it any time we see a
constraint with more variables than our current estimate. Since this
estimate for $k$ only increases over time, the analysis below will go
through unchanged. (We can assume that $k$ is a power of $2$---which
makes $\log k$ an integer; we will need that $k \geq 2$.)

\begin{lemma}
  \label{lem:bound-tj}
  For any constraint $\cin$, the number of primal updates $t_\cin\le 2 \log k$.
\end{lemma}
\begin{proof}
  Fix some $\cin$, and consider the value $i^*$ for which $d_{i^*\cin} =
  d_{m(\cin)}$. In each round the variable $x_{i^*} \gets 2x_{i^*} +
  1/(k\cdot a_{i^*\cin})$; hence after $t$ rounds its value will be at
  least $(2^t - 1)/(k\cdot a_{i^*\cin})$. So if we do $2 \log k$
  updates, this variable alone will satisfy the $\cin^{th}$ constraint.
\end{proof}

\begin{lemma}
  \label{lem:value-primal}
  The total increase in the value of the primal is at most $2\, t_\cin \,
  d_{m(\cin)}$.
\end{lemma}

\begin{proof}
  Consider a single update step that modifies primal variables from
  $x^{old}$ to $x^{new}$. In this step, the increase in each variable
  $i\in T_\cin$ is $\frac{d_{m(\cin)}}{d_{i\cin}} \cdot x_i^{old} +
  \frac{1}{k\cdot a_{i\cin}} \frac{d_{m(\cin)}}{d_{i\cin}}$. So the
  increase in the primal objective is:
$$\sum_{i\in T_\cin} c_i \cdot \left[ \frac{d_{m(\cin)}}{d_{i\cin}}  \cdot x_i^{old} + \frac{1}{k\cdot a_{i\cin}}
    \frac{d_{m(\cin)}}{d_{i\cin}} \right] \, = \, d_{m(\cin)}
  \sum_{i\in T_\cin} a_{i\cin} \cdot x_i^{old} + d_{m(\cin)}\cdot
  \frac{|T_\cin|}{k}\,  \le  \, 2\cdot d_{m(\cin)}$$ The inequality
  uses $|T_\cin|\le k$ and $\sum_{i\in T_\cin} a_{i\cin} \cdot x_i^{old}
  \le 1$ which is the reason an update was performed. The lemma now
  follows since $t_\cin$ is the number of update steps.
\end{proof}

To show approximate optimality, we want to change the dual variables so that the dual increase is (approximately) the
primal increase, and so that the dual remains (approximately) feasible. To achieve the first goal, we raise the newly
arriving dual variable, and to achieve the second we also decrease the ``first few'' dual variables in each dual
constraint where the new dual variable appears.
\begin{quote}
  For the $h^{th}$ primal constraint, let $d_{i\cin},
  d_{m(\cin)}, t_\cin$ be given by the primal update process.
  \begin{OneLiners}
  \item[(a)] Set $y_\cin \gets d_{m(\cin)} \cdot t_\cin$.
  \item[(b)] For each $i\in T_\cin$, do the following for dual constraint $\sum_{j} a_{ij} y_j \leq c_i$:
    \begin{OneLiners}
    \item[(i)] If $\sum_{j < \cin} a_{ij} y_j \leq (10 \log k)\, c_i$,
      do nothing; else
    \item[(ii)] Let $k_i < \cin$ be the largest index such that $\sum_{j
        \leq k_i} a_{ij} y_j \leq (5 \log k)\, c_i$; let $P_i = \{ j
      \leq k_i \mid i \in T_j\}$ be the indices of these first few dual variables
      that are active in the $i^{th}$ dual constraint. For all $j \in P_i$,
      \[ y_j^{new} \gets \left( 1 - \frac{d_{m(\cin)}}{d_{i\cin}}\right) \cdot y_j^{old}. \]
    \end{OneLiners}
  \end{OneLiners}
\end{quote}
Observe that the dual update process starts each dual variable $y_j$ off at some value $d_{m(j)} t_j$ and subsequently
\emph{only decreases} this dual variable, and that the dual variables remain non-negative.

\begin{lemma}
  \label{lem:inc-dual}
  When primal constraint $\cin$ arrives, the left-hand-side of each dual constraint
  $i$ increases due to the variable $y_\cin$ by $a_{i\cin} \cdot
  d_{m(\cin)} \cdot t_\cin \leq (2 \log k)\,c_i$.
\end{lemma}
\begin{proof}
  We set the initial value of the dual variable $y_\cin$ to $d_{m(\cin)} \cdot
  t_\cin$. By \lref[Lemma]{lem:bound-tj}, $t_\cin \leq 2 \log k$. By definition,
  $d_{m(\cin)} \leq c_i/a_{i\cin}$. Hence, for any $i \in T_\cin$, the increase
  in the left-hand-side of dual constraint $i$ is at most $a_{i\cin} \cdot (2 \log
  k)\, (c_i/a_{i\cin}) = (2 \log k)\,c_i$. This proves the lemma.
\end{proof}

\begin{lemma}
  \label{lem:well-def}
  When primal constraint $\cin$ arrives, if the dual update reaches
  step~b(ii) for some $i\in T_h$, then $k_i$ is well-defined and the set $P_i$ is non-empty;
  moreover, $\frac{\sum_{j \in P_i} a_{ij} y_j}{c_i \, \log k} \in
  [3,5]$.
\end{lemma}
\begin{proof}
  For each $j<\cin$ we have $y_j\le 2\log k\cdot d_{m(j)}$, since dual
  variable $y_j$ was initialized to $t_j d_{m(j)}\le 2\log k\cdot
  d_{m(j)}$ (by \lref[Lemma]{lem:bound-tj}) and subsequently never
  increased---so $a_{ij}\cdot y_j\le 2\log k\cdot d_{m(j)}\cdot
  a_{ij}\le 2\log k\cdot c_i$, using $d_{m(j)} \le d_{ij} =
  c_i/a_{ij}$. If the dual update reaches step~b(ii) then we have
  $\sum_{j < \cin} a_{ij} y_j > (10 \log k)\, c_i$, but each $j < \cin$
  contributes at most $2\log k \cdot c_i$, so $k_i$ is well-defined, and
  $P_i$ is non-empty. Moreover, by the choice of $k_i$, we have
  $\sum_{j \leq k_i +1} a_{ij} y_j > (5 \log k)\, c_i$, so $\sum_{j \leq
    k_i} a_{ij} y_j > (5 \log k)\, c_i - a_{i, k_i+1}\cdot y_{k_i+1} \ge
  (3\log k)\cdot c_i$, as claimed.
\end{proof}

\begin{lemma}
  \label{lem:approx-dual}
  After each dual update step, each dual constraint $i$ satisfies $\sum_j a_{ij} y_j
  \leq (12 \log k)\, c_i$. Hence the dual is $(12 \log k)$-feasible.
\end{lemma}
\begin{proof}
  Consider the dual update process when the primal constraint $\cin$
  arrives, and look at any dual constraint $i\in T_h$ (the other dual constraints are unaffected). If case~b(i) happens,
  then by \lref[Lemma]{lem:inc-dual} the left-hand-side of the constraint will be
  at most $(12 \log k)\, c_i$. Else, case~b(ii) happens. Each $y_j$ for
  $j \in P_i$ decreases by $y_j \cdot d_{m(\cin)}/d_{i\cin}$, and so
  the decrease in $\sum_{j \in P_i} a_{ij} y_j$ is at least $\sum_{j \in
    P_i} a_{ij} y_j \cdot (d_{m(\cin)}/d_{i\cin})$. Using
  \lref[Lemma]{lem:well-def}, this is at least
  \[ \frac{d_{m(\cin)}}{d_{i\cin}} \cdot c_i\,(3 \log k) =
  \frac{d_{m(\cin)}}{c_i/a_{i\cin}} \cdot c_i\,(3 \log k) = d_{m(\cin)}
  \cdot a_{i\cin}\cdot (3 \log k). \] But since the increase due to
  $y_\cin$ is at most $a_{i\cin} \cdot d_{m(\cin)}\, t_\cin \leq
  a_{i\cin} \cdot d_{m(\cin)}\cdot (2 \log k)$, there is no net increase
  in the LHS, so it remains at most  $(12 \log k)\, c_i$.
\end{proof}

\begin{lemma}
  \label{lem:value-dual}
  The net increase in the dual value due to handling primal constraint $\cin$
  is at least $\frac12\, d_{m(\cin)}\cdot t_\cin$.
\end{lemma}

\begin{proof}
  The increase in the dual value due to $y_\cin$ itself is
  $d_{m(\cin)}\cdot t_\cin$. What about the decrease in the other
  $y_j$'s? These decreases could happen due to any of the $k$ dual
  constraints $i \in T_\cin$, so let us focus on one such dual constraint
  $i$, which reads $\sum_{j: i \in T_j} a_{ij} y_j \leq c_i$. Now for $j
  < \cin$, define $\gamma_{ij} := \frac{y_j}{t_j\, d_{ij}}$. Since $y_j$
  was initially set to $t_j\, d_{m(j)} \leq t_j\, d_{ij}$ and
  subsequently never increased, we know that at this point in time,
  \begin{gather}
    \gamma_{ij} \quad \leq \quad \frac{d_{m(j)}}{d_{ij}} \quad \leq \quad 1.\label{eq:6}
  \end{gather}

The following claim, whose proof appears after this lemma, helps us bound the total dual decrease.
  \begin{cl}
    \label{clm:struct}
    If we are in case~b(ii) of the dual update, then $\sum_{j \in P_i}
    \frac{\gamma_{ij} t_j}{a_{ij}} \leq \frac{1}{2k} \cdot
    \frac{1}{a_{i\cin}}$.
  \end{cl}
Using this claim, we bound the loss in dual value caused by dual constraint~$i$:
  \begin{align*}
    \sum_{j \in P_i} \frac{d_{m(\cin)}}{d_{i\cin}} \cdot y_j & =
    \frac{d_{m(\cin)}}{d_{i\cin}} \cdot \sum_{j \in P_i} \gamma_{ij}
    \cdot t_j\, d_{ij}
     \,\, = \,\, \frac{d_{m(\cin)}}{c_i/a_{i\cin}} \cdot \sum_{j \in P_i}
    \gamma_{ij}
    \cdot t_j\, (c_i/a_{ij})\\
    & = d_{m(\cin)} \, a_{i\cin} \cdot \sum_{j \in P_i} \gamma_{ij}
    \cdot \frac{t_j}{a_{ij}} \leq_{(\text{\lref[Claim]{clm:struct}})} d_{m(\cin)} \, a_{i\cin}
    \cdot \frac{1}{2k} \cdot \frac{1}{a_{i\cin}} \,\, = \,\, \frac{d_{m(\cin)}}{2k}\,.
  \end{align*}
  Summing over the $|T_j| \leq k$ dual constraints affected, the total
  decrease is at most $\frac12 d_{m(\cin)} \leq \frac12 d_{m(\cin)}
  t_{\cin}$ (since there is no decrease when $t_{\cin} =
  0$). Subtracting from the increase of $d_{m(\cin)}\cdot t_\cin$ gives
  a net increase of at least $\frac12 d_{m(\cin)} t_{\cin}$, proving the lemma.
\end{proof}

\begin{proofof}{\lref[Claim]{clm:struct}}
  Consider the primal constraints $j$ such that $i \in T_j$: when they
  arrived, the value of primal variable $x_i$ may have increased. (In
  fact, if some primal constraint $j$ does not cause the primal
  variables to increase, $y_j$ is set to $0$ and never plays a role in
  the subsequent algorithm, so we will assume that for each primal
  constraint $j$ there is some increase and hence $t_j > 0$.)

  The first few among the constraints $j$ such that $i \in T_j$ lie in
  the set $P_i$: when $j \in P_i$ arrived, we added at least
  $\frac{1}{k\cdot a_{ij}} \frac{d_{m(j)}}{d_{ij}}$ to $x_i$'s
  value\footnote{More precisely, $x_i$ increased by at least
    $\frac{1}{k_j\cdot a_{ij}} \frac{d_{m(j)}}{d_{ij}}$ where $k_j\le k$
    was the estimate of the row-sparsity at the arrival of constraint
    $j$, and $k$ is the current row-sparsity estimate.}, and did so
  $t_j$ times.  Hence the value of $x_i$ after seeing the constraints in
  $P_i$ is at least $\sum_{j \in P_i} \frac{d_{m(j)} t_j}{k\cdot a_{ij}
    \cdot d_{ij}} \geq \sum_{j \in P_i} \frac{\gamma_{ij} t_j}{k\cdot
    a_{ij}}$, using~(\ref{eq:6}).

  If $\chi_i$ is the value of $x_i$ after seeing the constraints in $P_i$,
  and $\chi_i'$ is its value after seeing the rest of the constraints in
  $Q_i := (\{ j < \cin \mid i \in T_j \} \setminus P_i)$. Then
{\small  \begin{gather}
    \frac{\chi_i'}{\chi_i} \,\, \geq \,\, \prod_{j \in Q_i} \left(1 +\frac{d_{m(j)}}{d_{ij}}\right)^{t_j}
    \,\, \geq_{(\ref{eq:6})} \,\, \prod_{j \in Q_i} (1 + \gamma_{ij})^{t_j}
    \,\, \geq_{(\gamma_{ij} \leq 1)}\,\,     e^{\frac12 \sum_{j \in Q_i}
      \gamma_{ij} t_j} \,\, \geq \,\, 2k^2. \label{eq:7}
  \end{gather}}
  The last inequality uses the fact that $k \geq 2$, and that:
{\small $$\sum_{j\in Q_i} \gamma_{ij} t_j \,\, = \,\, \sum_{j \in Q_i} y_j/d_{ij} \,\, = \,\, \sum_{j \in Q_i}
\frac{y_j \cdot a_{ij}}{c_i} \,\, = \,\, \frac1{c_i} \left( \sum_{j<\cin} a_{ij}y_j - \sum_{j\in P_i} a_{ij}y_j\right)
\,\, > \,\, 5\log k,$$}where the inequality is because we are in case~b(ii) and $\sum_{j\in P_i} a_{ij}y_j \le (5\log
k)\cdot c_i$ by \lref[Lemma]{lem:well-def}.

Finally, when doing the primal/dual update steps for constraint $\cin$,
the value of $x_i$ just before this must have been $\chi_i' < 1/a_{i\cin}$
(otherwise constraint $\cin$ would have already been satisfied just by
variable $x_i$).  And $\chi_i$ is at least $\sum_{j \in P_i}
\frac{\gamma_{ij} t_j}{k\cdot a_{ij}}$, by the first calculations. And
$\chi_i'/\chi_i \geq 2k^2$ by~(\ref{eq:7}). Putting these together gives
\[ \sum_{j \in P_i} \frac{\gamma_{ij} t_j}{k\cdot a_{ij}} \leq
\frac{1}{2k^2} \cdot \frac{1}{a_{i\cin}}, \] and hence the claim.
\end{proofof}

\lref[Lemma]{lem:value-dual} and \lref[Lemma]{lem:value-primal} imply that
the dual increase is at least $1/4$ the primal increase, and
\lref[Lemma]{lem:approx-dual} implies we have an $O(\log k)$-feasible
dual, implying the following theorem:
\begin{theorem}
  Algorithm~I is an $O(\log k)$-competitive online algorithm for
  covering linear programs without upper-bound constraints, where $k$ is
  the row-sparsity of the constraint matrix.
\end{theorem}

\section{The Online Algorithm for CIPs}
We now want to solve CLPs with variable upper bounds, en route to solving general CIPs of the form (\ref{eq:1}).
However, it is well-known that when we have variable upper-bounds, the natural relaxation has a large integrality gap
even with a single constraint.\footnote{The trivial CIP $\min\{ x_1 \mid Mx_1 \geq 1\}$
  has integrality gap $M$, no upper bounds needed. However, if we
  truncate the $a_{ij}$s to be at most $1$ (which is the right-hand-side
  value), and we have no upper bound constraints, this gap
  disappears. Introducing upper bounds brings back large integrality
  gaps, as the example $\min\{ x_1 | x_1 + (1-\epsilon) x_2 \geq 1, x_2
  \leq 1\}$ shows, which has an integrality gap of $1/\epsilon$.}
Hence, Carr et al.~\cite{CFLP} suggested adding the \emph{knapsack
  cover} (KC) inequalities---defined below---to reduce the
integrality gap significantly. In this section, we first show how to extend Algorithm~I to get an $O(\log
k)$-competitive algorithm for the natural CLP relaxation (with upper bounds) where we also satisfy some suitable KC
inequalities. Next, we round (in an online fashion) such a fractional solution to get a randomized $O(\log \ell \cdot
\log k)$-competitive online algorithm for general $k$-row-sparse and $\ell$-column-sparse CIPs.

{\bf Knapsack Cover Inequalities.} Given a CIP of the form~(\ref{eq:1}), the KC-inequalities for a particular covering
constraint $\sum_{i \in [n]} a_{ij} x_i \geq 1$ are defined as follows: for {\em any} subset $H \sse [n]$ of variables,
the maximum possible contribution of the variables in $H$ to the constraint is $a_j(H) := \sum_{i \in H} a_{ij} u_i$,
and if $a_j(H)<1$ then at least a contribution of $1-a_j(H)$ must come from variables $[n]\setminus H$. Moreover, in
any {\em integral solution} $\x$, since each positive variable $x_i$ is at least one, we get the inequality:
\begin{equation} \label{eq:kc-ineq} \tsty \sum_{i\in [n]\setminus H}  \min\{a_{ij},\, 1-a_j(H)\} \cdot x_i \quad \ge \quad 1-a_j(H)
\end{equation}
Since~(\ref{eq:kc-ineq}) is not be true for an arbitrary \emph{fractional} solution satisfying \\$\sum_{i \in [n]}
a_{ij} x_i \geq 1$, we add this additional constraint to the LP, for each original constraint $j$ and $H \subseteq [n]$
where $a_j(H)<1$. There are exponentially many such KC-inequalities, and it is not known how to separate exactly over
these in poly-time\footnote{KC-inequalities can be separated in pseudo-polynomial time via a dynamic program for the
knapsack problem.}. But as in previous works~\cite{CFLP,KY05,BBN-stoc08-caching}, the randomized rounding algorithm
just needs us to enforce one specific KC-inequality for each constraint $j$---namely for the set $H_j :=\{i \in [n]
\mid x_i\ge \tau\cdot u_i\}$ with some suitable threshold $\tau > 0$. We call this the ``special'' KC-inequality for
constraint $j$.

\subsection{Fractional Solution with Upper Bounds and KC-inequalities}
\label{sec:box-kc}
\def\xo{\ensuremath{\overline{x}}\xspace}
\def\xbar{\ensuremath{\overline{\mathbf{x}}}\xspace}
\renewcommand{\a}{\alpha}
In extending Algorithm~I from the previous section to also handle ``box constraints'' (those of the form $0 \leq x_i
\leq u_i$), and the associated KC-inequalities, the high-level idea is to create a ``wrapper'' procedure around
Algorithm~I which ensures these new inequalities: when a constraint $\sum_{i \in T_j} a_{ij} x_i \geq 1$ arrives, we
start to apply the primal update step from Algorithm~I. Now if some variable $x_p$ gets ``close'' to its upper bound
$u_p$, we could then consider setting $x_p = u_p$, and feeding the new inequality $\sum_{i \in T_j  \setminus p} a_{ij}
x_i \geq 1 - a_{pj} u_p$ (or rather, a knapsack cover version of it) to Algorithm~I, and continuing. Implementing this
idea needs a little more work. For the rest of the discussion, $\tau\in (0,\frac12)$ is a threshold fixed later.

Suppose we want a solution to:
\[ \short{\tsty} (IP)\,\,\, \min \left\{ \sum_i c_i x_i \,\,\mid \,\,\sum_{i \in
    S_j} a_{ij} x_i \geq 1~~\forall j\in[m], \,\,0\le x_i \leq u_i, x_i \in \Z~~ \forall i\in[n], \right\}
\]
where constraint $j$ has $|S_j|\le k$ non-zero entries. The natural LP relaxation is:
\[ \short{\tsty} (P)\qquad \min \left\{ \sum_i c_i x_i \quad \mid \quad \sum_{i \in
    S_j} a_{ij} x_i \geq 1
~~\forall j\in[m], \quad 0\le x_i \leq u_i~~ \forall i\in[n] \right\}
\]
Algorithm~\ref{alg:box-clp} finds online a feasible fractional solution to this LP relaxation $(P)$, along with some
additional KC-inequalities. This algorithm maintains a vector $\x \in \R^n$ that need not be feasible for the covering
constraints in $(P)$. However $\x$ implicitly defines the ``real solution'' $\xbar \in \R^n$ as follows:
$$
\xo_i = \left\{
\begin{array}{ll}
x_i& \mbox{ if } x_i < \tau u_i\\
u_i & \mbox{ otherwise}
\end{array}\right.,\qquad \forall i\in[n]
$$
Let $\x^{(j)}$ and $\xbar^{(j)}$ denote the vectors immediately after the $j^{th}$ constraint to $(IP)$ has been
satisfied.
\begin{theorem}\label{thm:frac-box-kc}
  \lref[Algorithm]{alg:box-clp}, given the constraints of the CIP $(IP)$
  online, produces $\x$ (and hence $\xbar$) satisfying the following:
  \begin{OneLiners}
  \item[(i)] The solution $\xbar$ is feasible for $(P)$.
  \item[(ii)] The cost $\sum_{i=1}^n c_i\cdot x_i = O(\log k)\cdot \OPT_{IP}$.
  \item[(iii)] For each $j \in [m]$ let $H_j = \{i\in [n] \mid
    x^{(j)}_i\ge \tau\cdot u_i\}$ and $a_j(H_j)=\sum_{r\in H_j} a_{rj} u_r$. Then the solution $\x^{(j)}$ satisfies the
    KC-inequality corresponding to constraint $j$ with the set $H_j$, i.e., if $a_j(H_j)<1$ then:
    \[ \tsty \sum_{i\in S_j \setminus H_j}\,\,     \min\left\{ a_{ij},\, 1- a_j(H_j) \right\} \cdot x_i^{(j)} \quad
    \ge \quad 1- a_j(H_j). \]
  \end{OneLiners}
  Furthermore, the vectors $\x$ and $\xbar$ are non-decreasing over time.
\end{theorem}
Again, the value of row-sparsity $k$ is not required in advance---the algorithm just uses the current estimate as
before.


The solution $\xbar$ to $(P)$ is constructed by solving the (related) covering LP without upper-bounds---the
constraints here are defined by \lref[Algorithm]{alg:box-clp}.
\[ \short{\tsty} (P')\qquad \min \left\{ \sum_i c_i x_i \quad \mid \quad \sum_{i \in
    T_\cin} \a_{i\cin} x_i \geq 1 ~~\forall \cin\in[m'], \quad x_i \geq
  0~~ \forall i\in[n] \right\}
\]
At the beginning of the algorithm, $h = 0$. When the $j^{th}$ constraint for $(IP)$, namely $\sum_{i \in S_j} a_{ij}
x_i \geq 1$, arrives online, the algorithm generates (potentially several) constraints for $(P')$ based on it.
\lref[Claim]{cl:p-to-p'} shows these are all valid for~$(IP)$, so the optimal solution to $(P')$ is at most
$\OPT_{IP}$.

\begin{algorithm}[h!] \caption{Online covering with box
    constraints\label{alg:box-clp} } When constraint $j$ (i.e.,
  $\sum_{i\in S_j} a_{ij} \cdot x_i \geq 1$) arrives for $(P)$,

  \begin{algorithmic}[1]
    \STATE \label{line:1} set $\cin\gets \cin+1$, $t_\cin\gets 0$, $F_j\gets \{i\in S_j
    : x_i\ge  \tau u_i\}$, $T_\cin\gets S_j\setminus F_j$.

    \STATE \label{line:2} set $b\gets 1-\sum_{i\in F_j} a_{ij}u_i$, and $\a_{i\cin} \gets
    \min\left\{ 1,\, \frac{a_{ij}}{b}\right\},\, \forall i\in T_\cin$,
    and $\a_{i\cin} = 0,\, \forall i\not\in T_\cin$.

    \STATE if $b>0$ then generate constraint $\sum_{i \in T_\cin} \a_{i\cin} x_i \geq
    1$ for $(P')$ else \textbf{halt}. \\
    // \emph{If $b\le 0$ then constraint $j$ to $(P)$ satisfied}

    \WHILE{ $(\sum_{i\in T_h} \a_{ih} \cdot x_i < 1)$} \label{line:3}

    \STATE \label{line:4} \COMMENT{{\em start primal-update process for
        $\cin^{th}$ constraint $(\sum_{i\in T_\cin} \a_{i\cin} \cdot x_i
        \geq 1)$ to $(P')$.}}

    \STATE if $T_\cin = \emptyset$, return \textsc{infeasible}.

    \STATE \label{line:5} define $d_{i\cin} := \frac{c_i}{\a_{i\cin}}$ for all $i \in
    [n]$, and $d_{m(\cin)} := \min_i d_{i\cin} := \min_{i \in T_\cin}
    d_{i\cin}$.

    \STATE \label{line:6} define $\delta\le 1$ to the maximum value in $(0,1]$ so that:
    \[ \max_{i\in T_\cin} \,\, \left\{ \frac{1}{u_i} \left[ \left( 1+ \delta\cdot
        \frac{d_{m(\cin)}}{d_{i\cin}} \right) x_i^{old} ~+~
      \frac{\delta}{k\cdot \a_{i\cin}}
      \frac{d_{m(\cin)}}{d_{i\cin}} \right] \right\} \quad \le \quad \tau\]

    \STATE \label{line:7} perform an update step {\em for constraint $\cin$} as:
    \[ x_i^{new} \gets \left( 1+ \delta\cdot
      \frac{d_{m(\cin)}}{d_{i\cin}} \right) x_i^{old} ~~+~~
    \frac{\delta}{k\cdot \a_{i\cin}} \frac{d_{m(\cin)}}{d_{i\cin}},
    \qquad \forall i\in T_\cin. \]

    \STATE \label{line:9} set $t_\cin\gets t_\cin+\delta$.

    \STATE \label{line:9} let $F'_\cin\gets \{i\in T_\cin : x_i = \tau u_i\}$ and $F_j\gets
    F_j\bigcup F'_\cin$.  \quad //{\em $\xo_i= u_i \iff i\in F_j$.}

    \IF{$(F'_\cin \ne \emptyset)$}

    \STATE \label{line:11} \COMMENT{{\em constraint $\cin$ to $(P')$ is
        deemed to be satisfied
      and new constraint $\cin+1$ is generated.}}

    \STATE \label{line:12} set $\cin\gets \cin+1$, $t_\cin\gets 0$, and $T_\cin \gets
    S_j\setminus F_j$.

    \STATE \label{line:13} set $b\gets 1-\sum_{i\in F_j} a_{ij} u_i$,  $\a_{i\cin} =
    \min\left\{ 1,\, \frac{a_{ij}}{b}\right\},\, \forall i\in T_\cin$
 and $\a_{i\cin} = 0,\,\forall i\not\in T_\cin$.

    \STATE if $b>0$ generate constraint $\sum_{i \in T_\cin} \a_{i\cin} x_i \geq
    1$ for $(P')$; else \textbf{halt}. \\
    // \emph{If $b\le 0$ then constraint $j$ to $(P)$ satisfied}

    \ENDIF

    \ENDWHILE \quad  \COMMENT{{\em constraint $j$ to $(P)$ is now satisfied.}}
\end{algorithmic}
\end{algorithm}

Clearly $\xbar \in [0, \mathbf{u}]$; it is feasible for $(P)$ because (a)~we increase variables until the condition in
\lref[line]{line:3} is satisfied, and (b)~if $h$ denotes the current constraint to $(P')$ at any point in the
while-loop, the following invariant holds: \full{\begin{quote}} {\em Solution $X$ satisfies constraint $h$ to
    $(P')$, i.e. $\sum_i \alpha_{ih}\cdot X_i\ge 1$, \full{ \\
      $\implies$\quad}\short{then} $\overline{X}$ satisfies constraint
    $j$ to $(P)$, i.e. $\sum_i a_{ij}\cdot \overline{X}_i\ge 1$.}
  \full{\end{quote}}
By construction $\x$ and $\xbar$ are non-decreasing over the run of the algorithm. Finally, for property~(iii), note
that the condition of the while loop captures this very KC inequality since $T_h=\{i\in S_j : x_i<\tau\cdot u_i\}$ at
all times.

To show property~(ii), we use a primal-dual analysis as in Section~\ref{sec:clps}: we will show how to maintain an
$O(\log k)$-feasible dual $y$ for $(P')$, so that $\mathbf{c} \cdot \x$ is at most $O(1)$ times the dual objective
$\sum_{h \in [m']} y_h$. This means $\mathbf{c} \cdot \x \leq O(\log k) \OPT_{P'} \leq O(\log k) \cdot \OPT_{IP}$, with
the last inequality following from \lref[Claim]{cl:p-to-p'} below.

\begin{cl}\label{cl:p-to-p'}
  The optimal value for the LP $(P')$ is at most
  $\OPT_{IP}$, the optimum integer solution to $(IP)$.
\end{cl}
\begin{proof}
  We claim that every inequality in $(P')$ can be obtained as a
  KC-inequality generated for $(IP)$. Indeed, consider the $\cin^{th}$
  constraint $\sum_{i \in T_\cin} \a_{i\cin} x_i \geq 1$ added to
  $(P')$, say due to the $j^{th}$ constraint $\sum_{i\in S_j} a_{ij}
  \cdot x_i \ge 1$ of $(IP)$. Here $T_\cin = S_j\setminus F_j$ for some
  $F_j\sse S_j$, and $\a_{i\cin} = \min\left\{ 1,\,
    \frac{a_{ij}}{b}\right\}$ for $i\in T_\cin$ with $b=1-\sum_{r\in
    F_j} a_{rj}\cdot u_r >0$.  In other words, the $\cin^{th}$
  constraint to $(P')$ reads
  $$ \short{\tsty} \sum_{i\in S_j\setminus F_j}  \min\left\{1-\sum_{r\in F_j} a_{rj}\cdot u_r,\,\,\, a_{ij}\right\}\cdot x_i \quad \ge\quad   1-\sum_{r\in F_j} a_{rj}\cdot u_r,$$
  which is the KC-inequality from the $j^{th}$ constraint of $(IP)$ with
  fixed set $F_j$. Now since all KC-inequalities are valid for any
  integral solution to $(IP)$, the original claim follows.
\end{proof}

Now to show how to maintain the approximate dual solution for $(P')$, and bound the cost of the primal update in terms
of this dual cost. The dual of $(P')$ is:
\[ \short{\tsty} (D') \qquad \max \left\{ \sum_{\cin=1}^{m'} y_\cin \quad \mid \quad
  \sum_{\cin: i \in T_\cin} \a_{i\cin} \cdot y_\cin \leq c_i ~~\forall i\in[n], \quad
  y_\cin \geq 0~~ \forall j\in[m'] \right\}
\]
The dual update process is similar to that in \lref[Section]{sec:clps}. When constraint $\cin$ to $(P')$ is deemed
satisfied in \lref[line]{line:11}, update dual $y$ as follows:
\begin{quote}
  Let $d_{i\cin}, d_{m(\cin)}, t_\cin$ be as defined in
  \lref[Algorithm]{alg:box-clp}.
  \begin{OneLiners}
  \item[(a)] Set $y_\cin \gets d_{m(\cin)} \cdot t_\cin$.
  \item[(b)] For each dual constraint $i$ s.t.\  $i\in T_\cin$ (i.e., $\sum_{l: i \in T_{l}}
    \a_{il} y_l \leq c_i$), do the following:
    \begin{OneLiners}
    \item[(i)] If $\sum_{l < \cin} \a_{il} y_l \leq (10 \log k)\, c_i$,
      do nothing; else
    \item[(ii)] Let $k_i < \cin$ be the largest index such that $\sum_{l
        \leq k_i} \a_{il} y_l \leq (5 \log k)\, c_i$; let $P_i = \{ l
      \leq k_i \mid i \in T_l\}$ be the indices of these first few dual
      variables active in dual constraint $i$. For all $l \in P_i$, set
      \[ y_l^{new} \gets \left( 1 - \min\{1,\,t_\cin\} \cdot
        \frac{d_{m(\cin)}}{d_{i\cin}}\right) \cdot y_l^{old}. \]
    \end{OneLiners}
  \end{OneLiners}
\end{quote}
The only difference from Section~\ref{sec:clps} is to change $\left( 1 -
  \frac{d_{m(\cin)}}{d_{i\cin}}\right)$ to \\
$\left( 1 - \min\{1,t_\cin\}
  \frac{d_{m(\cin)}}{d_{i\cin}}\right)$; this is because maintaining
$x_i \leq \tau u_i$ required us to be cautious and introduce the damping factor of $\delta \in (0,1]$ in the primal
update, hence $t_\cin$ could be much smaller than one. Here too, each $y_\cin$ starts off at $d_{m(j)} t_\cin$, and
only decreases thereafter. Similar to \lref[Lemmas]{lem:bound-tj} and~\ref{lem:value-primal}, we get:
\begin{lemma}
  \label{lem:bound-tj-b}
  For any constraint $\cin$ to $(P')$, the value $t_\cin\le 2 \log k$.
\end{lemma}
\begin{proof}
  (Sketch) Each time $t_\cin$ increases by $1$, the process behaves as
  before, so if we perform a primal increase step then $t_\cin$ is an
  integer strictly less than $2\log k$ (itself an integer since we
  assumed $k$ is a power of $2$). Also, the first time that $t_\cin$
  increases by $\delta<1$, the algorithm adds at least one variable to
  $F'_\cin$, fixes $t_\cin$ and moves on to a new constraint $\cin+1$.
\end{proof}

In the rest of the proof, we omit details that are repeated from Section~\ref{sec:clps}, and only point out
differences, if any.

\begin{lemma}
  \label{lem:value-primal-b}
  The total increase in $\sum_{i\in[n]} c_i\cdot x_i$ due to updates for
  constraint $\cin$ is at most $2\, t_\cin \, d_{m(\cin)}$.
\end{lemma}
\ignore{\begin{proof}
  Consider a single update step that modifies primal variables from
  $x^{old}$ to $x^{new}$. In this step, the increase in each variable
  $i\in T_\cin$ is $\delta \frac{d_{m(\cin)}}{d_{i\cin}} \cdot x_i^{old} +
  \frac{\delta}{k\cdot \a_{i\cin}} \frac{d_{m(\cin)}}{d_{i\cin}}$. So the
  increase in the primal objective is:
  $$\delta \sum_{i\in T_\cin} c_i \cdot \left[ \frac{d_{m(\cin)}}{d_{i\cin}}
    \cdot x_i^{old} + \frac{1}{k\cdot \a_{i\cin}}
    \frac{d_{m(\cin)}}{d_{i\cin}} \right] \,\, = \,\, \delta\,
  d_{m(\cin)} \sum_{i\in T_\cin} \a_{i\cin} \cdot x_i^{old} +
  d_{m(\cin)}\cdot \frac{|T_\cin|}{k}\,\, < \,\, 2\delta\cdot
  d_{m(\cin)},$$ where we used $d_{i\cin} = c_i/\a_{i\cin}$ and
  $|T_\cin|\le |S_j|\le k$. Moreover, we used that $\sum_{i\in T_\cin}
  \a_{i\cin} \cdot x_i^{old} < 1$, since the termination condition of
  the while loop was not met by $\x^{old}$.
\end{proof}
}

\begin{lemma}
  \label{lem:inc-dual-b}
  In the dual update for constraint $\cin$ to $(P')$, variable $y_\cin$
  increases the left-hand-side of each dual constraint $i$ by $\a_{i\cin} \cdot
  d_{m(\cin)} \cdot t_\cin \leq (2 \log k)\cdot c_i$.
\end{lemma}
\ignore{
\begin{proof}
  We set the initial value of the dual variable $y_\cin$ to $d_{m(\cin)}
  \cdot t_\cin$. By \lref[Lemma]{lem:bound-tj-b}, $t_\cin \leq 2 \log
  k$. By definition, $d_{m(\cin)} \leq c_i/\a_{i\cin}$. Hence, for any $i
  \in T_\cin$, the increase in the left-hand-side of dual constraint $i$
  is at most $\a_{i\cin} \cdot (2 \log k)\, (c_i/\a_{i\cin}) = (2 \log
  k)\,c_i$. This proves the lemma.
\end{proof}
}

\begin{lemma}
  \label{lem:well-def-b}
  If the dual update for constraint $\cin$ to $(P')$ reaches step~b(ii),
  then $k_i$ is well-defined and the set $P_i$ is non-empty; moreover,
  $\frac{\sum_{l \in P_i} \a_{il} y_l}{c_i \, \log k} \in [3\ldots 5]$.
\end{lemma}

\ignore{ \begin{proof}
  For each $l<\cin$ we have $y_l\le 2\log k\cdot d_{m(l)}$, since dual
  variable $y_l$ was initialized to $t_l d_{m(l)}\le 2\log k\cdot
  d_{m(l)}$ (by \lref[Lemma]{lem:bound-tj-b}) and subsequently never
  increased---so $\a_{il}\cdot y_l\le 2\log k\cdot d_{m(l)}\cdot
  \a_{il}\le 2\log k\cdot c_i$, using $d_{m(l)} \le d_{il} =
  c_i/\a_{il}$. If the dual update reaches step~b(ii) then we have
  $\sum_{l < \cin} \a_{il} y_l > (10 \log k)\, c_i$, but each $l < \cin$
  contributes at most $2\log k \cdot c_i$, so $k_i$ is well-defined, and
  $P_i$ is non-empty. Moreover, by the choice of $k_i$, we have
  $\sum_{l \leq k_i +1} \a_{il} y_l > (5 \log k)\, c_i$, so $\sum_{l \leq
    k_i} \a_{il} y_l > (5 \log k)\, c_i - \a_{i, k_i+1}\cdot y_{k_i+1} \ge
  (3\log k)\cdot c_i$, as claimed. \agnote{can omit this proof.}
\end{proof}
}

\begin{lemma}
  \label{lem:approx-dual-b}
  After each dual update step, the dual is $(12 \log k)$-feasible; i.e. each dual constraint $\sum_l \a_{il} y_l
  \leq (12 \log k)\, c_i$.
\end{lemma}
\begin{proof}
  As in the proof of \lref[Lemma]{lem:approx-dual}, consider the
  update due to constraint $\cin$ to $(P')$ and the $i^{th}$ dual
  constraint for some $i \in T_\cin$. If we are in case~b(i),
  \lref[Lemma]{lem:inc-dual-b} implies that $\sum_l \a_{il} y_l \leq (10
  \log k)c_i + (2 \log k)c_i$. For case~b(ii), the decrease in the left-hand-side
  $\sum_{l \in P_i} \a_{il} y_l$ of constraint $i$ is at least
  $\min\{1,t_\cin\}\cdot \sum_{l \in P_i} \a_{il} y_l \cdot
  (d_{m(\cin)}/d_{i\cin})$. By \lref[Lemma]{lem:well-def-b} the sum
  $\sum_{l \in P_i} \a_{il} y_l \geq c_i\,(3 \log k)$ and hence the
  reduction in the left-hand-side of dual constraint $i$ is at least
  \[ \min\{3 \log k,\,t_\cin\}\cdot \frac{d_{m(\cin)}}{d_{i\cin}} \cdot c_i\quad = \quad d_{m(\cin)}
  \cdot \a_{i\cin}\cdot \min\{3 \log k,\,t_\cin\} \quad \ge \quad d_{m(\cin)}
  \cdot \a_{i\cin}\cdot t_\cin. \] The inequality uses \lref[Lemma]{lem:bound-tj-b}. Combined with \lref[Lemma]{lem:inc-dual-b}
  it follows that there is no net increase in the left-hand-side. Hence we can maintain the invariant that it is at most
  $(12 \log k)\, c_i$.
\end{proof}

\begin{lemma}
  \label{lem:value-dual-b}
  The net increase in dual value due to handling constraint $\cin$ to $(P')$
  is at least $\frac12\, d_{m(\cin)}\cdot t_\cin$.
\end{lemma}
\begin{proof}
  The increase in the dual value due to $y_\cin$ is $d_{m(\cin)}\cdot
  t_\cin$. As in \lref[Lemma]{lem:value-dual}, let us bound the decrease
  in the other $y_l$'s. Consider any of the $k$ dual constraints $i \in
  T_\cin$. Again define $\gamma_{il} := \frac{y_l}{t_l\, d_{il}}$ for $l
  < \cin$; since $y_l$ started off at $t_l \cdot d_{m(l)}$ and never
  increased, we have $\gamma_{il} \leq d_{m(l)}/d_{il} \leq 1$. Again,
  as in \lref[Claim]{clm:struct}:
  \begin{cl}
    \label{clm:struct-b}
    If we are in case~b(ii) of the dual update, then $\sum_{l \in P_i}
    \frac{\gamma_{il} t_l}{\a_{il}} \leq \frac{1}{2k} \cdot
    \frac{1}{\a_{i\cin}}$.
  \end{cl}
   Using calculations  as in \lref[Lemma]{lem:value-dual},
  the decrease in dual objective due to dual constraint~$i$ is:
  $$
  \min\{1,\,t_\cin\}\cdot \sum_{l \in P_i} \frac{d_{m(\cin)}}{d_{i\cin}}
  \cdot y_l \quad \le \quad \frac1{2k} \, d_{m(\cin)}\cdot
  \min\{1,\,t_\cin\}\quad \le \quad \frac1{2k} \, d_{m(\cin)}\cdot
  t_\cin.$$ Since there are $|T_\cin| \leq k$ dual constraints we have
  to consider, the total decrease is at most $\frac12
  d_{m(\cin)}\,t_\cin$. Subtracting this from the total increase of
  $d_{m(\cin)}\cdot t_\cin$ gives the lemma.
\end{proof}

Comparing \lref[Lemma]{lem:value-dual-b} with \lref[Lemma]{lem:value-primal-b}, while handling the $h^{th}$ constraint
in $(P')$ the increase in the dual objective function is at least $1/4$ of the increase in the primal objective
function $\mathbf{c} \cdot \x$. And \lref[Lemma]{lem:approx-dual-b} tells us that $y$ is an $O(\log k)$-feasible dual
to $(P')$. Hence:
$$\mathbf{c} \cdot \x \quad \le \quad 4 ({\bf 1} \cdot  \mathbf{y}) \quad \le_{weak~duality} \,\, O(\log k)\cdot
\OPT_{P'}\quad \le_{\lref[Claim]{cl:p-to-p'}}\,\, O(\log k)\cdot \OPT_{IP}.$$ This completes the proof of
property~(iii) in \lref[Theorem]{thm:frac-box-kc}.


\subsection{Online Rounding}
\label{sec:rounding}

We now complete the algorithm for CIPs by showing how to round the online fractional solution generated by
\lref[Theorem]{thm:frac-box-kc} also in an online fashion. This rounding algorithm also does randomized rounding on the
incremental change like in~\cite{AAABN03}, but to get a loss of $O(\log \ell)$ instead $O(\log m)$, we use the method
of randomized rounding with alterations~\cite{AS-book,Srinivasan01}. Recall $\ell\le m$ is the column-sparsity of the
constraint matrix $A$---the maximum number of constraints any variable $x_i$ participates in. (The $O(\log \ell)$ bound
for offline CIPs given by~\cite{Srinivasan06,KY05} uses a derandomization of the Lov\'asz Local Lemma via pessimistic
estimators, and is not applicable in the online setting.)

Given that the constraints of a CIP arrive online, we run
\lref[Algorithm]{alg:box-clp} to maintain vectors $\x$ and $\xbar$ with
properties guaranteed by \lref[Theorem]{thm:frac-box-kc}. For this
section, we set the threshold $\tau$ to $\frac18\cdot \frac1{\log
  \ell}$. Before any constraints arrive, pick a uniformly random value
$\rho_i\in [0,1]$ for each variable $i \in [n]$---this is the only
randomness used by the algorithm. We will maintain an integer solution
$\textbf{X} \in \Z^n_{\geq 0}$; again let $\textbf{X}^{(j)}$ denote this
solution right after primal constraint $j$ has been satisfied. We start
off with $\textbf{X}^{(0)} = \textbf{0}$.  When the $j^{th}$ constraint
arrives and the (fractional) $x_i$ values have been increased in
response to this constraint, we do the following.

\begin{enumerate}
\item Define the ``rounded unaltered'' solution:
$$
Z_i = \left\{
  \begin{array}{ll}
    0 & \mbox{ if } x_i < \tau \rho_i\\
    \lceil x_i/\tau \rceil & \mbox{ if } \tau \rho_i \le x_i < \tau u_i \\
    u_i &\mbox{ if } x_i \ge \tau u_i
  \end{array}\right.,\qquad \forall i\in[n].
$$

\item \emph{Maintain monotonicity.} Define:
\[ X_i^{new} = \max\{ X_i^{(j-1)},\; Z_i \},\qquad \forall i\in[n]. \]
Observe that this rounding ensures that $X_i \in \{0,1,\ldots, u_i\}$
for all $i \in [n]$.

\item \emph{Perform potential alterations.} If we are unlucky and the
  arriving constraint $j$ is not satisfied by $\textbf{X}^{new}$, we
  increase $\textbf{X}^{new}$ to cover this constraint $j$ as
  follows. Let $H_j := \{ i \in [n] \mid x^{(j)}_i \ge \tau \cdot u_i
  \}$ be the frozen variables in the fractional solution; note that
  $Z_i=u_i$ for all $i\in H_j$, so these variables cannot be
  increased. Recall that $a_j(H_j) : = \sum_{r\in H_j} a_{rj}\cdot
  u_r$. Since constraint $j$ is not satisfied, $a_j(H_j)<1$ and the
  algorithm performs the following alteration for constraint
  $j$. Consider the residual constraint on variables $[n]\setminus H_j$
  after applying the KC-inequality on $H_j$, i.e.
  $$\sum_{i\in [n]\setminus H_j} \min\{ a_{ij},\, 1-a_j(H_j)\}\cdot w_i
  \quad \ge \quad 1-a_j(H_j). $$
  Set $\overline{a}_{ij} = \min\left\{ 1,\,
    \frac{a_{ij}}{1-a_j(H_j)}\right\}$ for all $i\in [n]\setminus
  H_j$. Consider the following covering knapsack problem:
{\small  \begin{alignat}{10}
    \min \tsty \quad \sum_{i \in [n]\setminus H_j} \,\,  c_i \cdot w_i &
    \tag{$IP_K$} \label{eq:minknap1} \\
    \tsty \mbox{subject to:} \quad  \sum_{i \in [n]\setminus H_j}  \overline{a}_{ij} \cdot w_i \geq 1 & \notag\\
    0\le w_i \le u_i, & \qquad \forall i \in [n]\setminus H_j \notag \\
    w_i \in \mathbb{Z}, & \qquad \forall i \in [n]\setminus
    H_j \notag
\end{alignat}}
  Note that there is only one covering constraint in this problem. Let
  $W$ denote an approximately optimal integral solution obtained by the
  natural greedy algorithm. It is clear that $W$ satisfies the residual
  constraint $j$ on variables $[n]\setminus H_j$. Define
  $\textbf{X}^{(j)}$ as follows.
  $$
  X^{(j)}_i = \left\{
    \begin{array}{ll}
      X^{new}_i & \mbox{ for  } i \in H_i \\
      \max\left\{  X^{new}_i,\, W_i\right\}  & \mbox{ for } i \in [n]\setminus H_j
    \end{array}\right.
  $$
\end{enumerate}

This completes the description of the algorithm. By construction, it
outputs a feasible integral solution to the constraints so far, so it
remains to bound its expected cost.

\textbf{Remark:}
This algorithm does not require knowledge of the final column-sparsity
$\ell$ in advance. At each step, we use the current value of
$\ell$. Notice that this only affects $\tau$ and the definition of
$\textbf{Z}$. However, for fixed values of $x_i$ and $\rho_i$ (any
$i\in[n]$) the value of $Z_i$ is non-decreasing with $\ell$: so vector
$\textbf{Z}$ is monotone over time (since $\ell$ is non-decreasing). We
also require a slightly more general version of
\lref[Theorem]{thm:frac-box-kc} where we have multiple thresholds
$\tau_1\le \tau_2\le \cdots \le \tau_m$ and replace $\tau$ by $\tau_j$
in condition~(iii). This extension is straightforward and details are
omitted.

{\bf Cost of $\textbf{Z}$.} Consider the rounding algorithm immediately after all $m$ constraints have been satisfied.
If $x_i/\tau \in [0,1]$, then $\E[Z_i] = \Pr[ \rho_i \leq x_i/\tau] = x_i/\tau$; if $x_i/\tau \geq 1$, then $Z_i \le
\lceil x_i/\tau \rceil \leq 2x_i/\tau$ with probability $1$. Hence:
\[ \short{\tsty} \E\left[ \sum_{i=1}^n c_i \cdot Z_i \right] \quad \leq \quad (2/\tau) \sum c_i x_i \quad = \quad O(\log k \cdot \log
\ell) \cdot \OPT_{IP}, \] where we use $1/\tau = O(\log \ell)$, and \lref[Theorem]{thm:frac-box-kc}(ii) to bound
$\sum_i c_ix_i$.

{\bf Cost of $\textbf{X}-\textbf{Z}$.} To account for $\textbf{X}-\textbf{Z}$, we need to bound the expected cost of
any alterations. In the sequel, let $\ell_j$, $k_j$ and $\tau_j$ denote the respective values of $\ell$, $k$ and $\tau$
at the arrival of constraint $j$. When $j$ is clear from context we will drop the subscript.

Recall that $H_j := \{ i \in [n] \mid x^{(j)}_i \ge \tau_j\cdot u_i\}$ are the frozen variables in the fractional
solution after handling constraint $j$, and note $Z_i=u_i$ for $i \in H_j$. Define $A_j := \{ i \in [n] \mid x^{(j)}_i
< \tau_j \}$. Note that the randomness only plays a role in the values of $\{ Z_i \mid i \in A_j\}$, since all
variables in $[n] \setminus A_j$ deterministically are set to $Z_i = \min\left\{
  \lceil x^{(j)}_i/\tau_j \rceil,\, u_i\right\}$. Let $\Ev_j$ denote
the event that an alteration was performed for constraint $j$. The event
$\Ev_j$ occurs exactly when $\sum_{i \in [n]} a_{ij} \cdot X^{new}_i <
1$. Since variables $r\in H_j$ have $X^{new}_r=Z_r=u_r$ with probability
$1$, event $\Ev_j$ is the same as $a_j(H_j) <1$ (which is a
deterministic condition) and $\sum_{i \in [n] \setminus H_j} a_{ij}
\cdot X^{new}_i < 1 - a_j(H_j)$.

\begin{lemma}
  \label{lem:prob-alter}
  The probability of an alteration for constraint $j$ is $\Pr[\Ev_j] \leq \frac1{\ell_j^2}$.
\end{lemma}

\begin{proof}
Let $b = 1 - a_j(H_j)$, for $\Ev_j$ to occur we have $b>0$. Set $\overline{a}_{ij} = \min\{ a_{ij}/b, 1 \}$ for
$i\in[n]\setminus H_j$. Now since $\textbf{Z} \leq \textbf{X}$ and both are integer-valued, $\Pr[\Ev_j]$ {\small \[  =
\Pr\left[ \sum_{i \in [n] \setminus H_j} a_{ij} \cdot X^{new}_i < b \right] \leq \Pr\left[ \sum_{i \in [n] \setminus
H_j} a_{ij} \cdot Z_i < b \right] = \Pr\left[ \sum_{i \in [n] \setminus H_j} \overline{a}_{ij} \cdot Z_i < 1 \right]~.
\]}
\lref[Theorem]{thm:frac-box-kc}(iii) guarantees that $\sum_{i \in [n] \setminus H_j} \overline{a}_{ij} \cdot x^{(j)}_i
\geq 1$. Among $i\in [n]\setminus H_j$,
\begin{OneLiners}
\item $Z_i = \lceil x^{(j)}_i/\tau\rceil$ deterministically for $i \in
  [n] \setminus (H_j \cup A_j)$, and
\item $Z_i \in \{0,1\}$ with $\E[ Z_i ] = x^{(j)}_i/\tau$ independently for $i\in A_j$.
\end{OneLiners}
So $\E\left[ \sum_{i \in [n] \setminus H_j} \overline{a}_{ij} \cdot Z_i\right] \ge \frac1\tau$. Now Chernoff bound
implies for a collection of $[0,1]$-valued independent random variables, that the probability of their sum being less
than $\tau = 1/(8 \log \ell_j)$ times their expectation is at most $1/\ell_j^2$.
\end{proof}

\begin{lemma}
Conditioned on $\Ev_j$, the cost of incrementing $\textbf{X}^{new}$ to $\textbf{X}^{(j)}$ is at most $36\,\sum_{i\in
S_j}c_i\cdot x_i^{(j)}$; here $S_j\sse [n]$ are the non-zero columns in  constraint $j$.
\end{lemma}
\begin{proof}
The fractional solution $\x^{(j)}$ satisfies the KC inequality for set
  $H_j$, by \lref[Theorem]{thm:frac-box-kc}(iv). In particular, setting
  $w'_i=x^{(j)}_i$ for $i\in S_j\setminus H_j$ (and zero otherwise)
  gives a feasible \emph{fractional} solution to the LP relaxation of
  the covering knapsack subproblem~\eqref{eq:minknap1}. It suffices to
  show that the greedy integral solution $W$ to~\eqref{eq:minknap1}
  costs $36 \sum_{i\in S_j}c_i\cdot w'_i$. It is crucial that
  $w'_i\le \tau\cdot u_i<u_i/2$ for all $i\in[n]\setminus H_j$, as
  in general the integrality gap due to relaxing~\eqref{eq:minknap1} is
  unbounded.

  The greedy algorithm orders columns $i\in [n]\setminus H_j$ in
  non-decreasing $c_i/\overline{a}_{ij}$ order, and increases $W_i$
  variables integrally (up to their $u_i$s) until $\sum_i
  \overline{a}_{ij}\cdot W_i\ge 1$. Since all $\overline{a}_{ij}\le 1$,
  it is easy to show that this algorithm achieves a 2-approximation for
  covering knapsack~\eqref{eq:minknap1}.

  To complete the proof, we show the optimal integral solution
  to~(\ref{eq:minknap1}) costs at most $18 \sum_{i\in S_j}c_i\cdot
  w'_i$: we give a rounding algorithm to obtain an integral solution
  $W'$ from $w'$ with only a factor $18$ increase in cost.  Set $W'_i
  \sim \text{Binom}(u_i,\, 2w'_i/u_i)$ for all $i\in [n]\setminus
  H_j$---this definition is valid since $w'_i\le u_i/2$. Clearly $W'$
  always satisfies the upper bounds $u_i$ and has expected cost $2\,
  \textbf{c}\cdot \textbf{w'}$. Moreover, each $W_i'$ is a binomial
  r.v.\ and $\overline{a}_{ij} \leq 1$, so $\sum_i
  \overline{a}_{ij}\cdot W'_i$ can be viewed as a sum of independent
  $[0,1]$-valued random variables. The expectation $\E\left[ \sum_i
    \overline{a}_{ij}\cdot W'_i\right]\ge 2$, so a Chernoff bound gives
  $\Pr\left[ \sum_i \overline{a}_{ij}\cdot W'_i <1\right]\le 8/9$. Using
  Markov's inequality, $\Pr\left[ \textbf{c}\cdot \textbf{W'} > 18\,
    \textbf{c}\cdot \textbf{w'}\right]<1/9$. So with positive
  probability, $W'$ satisfies~\eqref{eq:minknap1} and costs at most
  $18\, \textbf{c}\cdot \textbf{w'}$, showing that $\Opt$\eqref{eq:minknap1} is at most this cost.
\end{proof}

 Thus the total expected cost of alterations after $m$ constraints is:
{\small \begin{eqnarray*} &&\sum_{j=1}^m \Pr[\Ev_j]\cdot 36\,\sum_{i\in S_j}c_i\cdot x_i^{(j)} \,\, \le\,\,  36\,
\sum_{j=1}^m
\frac1{\ell_j^2}\cdot \sum_{i\in S_j}c_i\cdot x_i^{(j)} \,\,  \le\,\,  36\, \sum_{i=1}^n c_i\cdot x_i^{(m)} \,\left(\sum_{j:i\in S_j} \frac1{\ell_j^2}\right)\\
&& \qquad \le \,\, 36\, \sum_{i=1}^n c_i\cdot x_i^{(m)} \,\left(\frac1{1^2} +
  \frac1{2^2} +\cdots + \frac1{\ell_j^2}\right)\,\, \le \,\, 9\pi^2\, \sum_{i=1}^n c_i\cdot x_i^{(m)}.
\end{eqnarray*}}
The second inequality uses the monotonicity of the fractional solution $\x$, and the third inequality uses that for any
$i\in[n]$, the value $\ell_j$ is at least $q$ upon arrival of the $q^{th}$ constraint containing variable $i$.

Combining the expected cost of $O( \mathbf{c} \cdot \x)$ for the
alterations with the expected cost of $O(\log \ell)\cdot
(\textbf{c}\cdot \textbf{x})$ for the initial rounding, and
\lref[Theorem]{thm:frac-box-kc}(ii), we get the main result for
this section:

\begin{theorem}
  There is an $O(\log k\cdot \log \ell)$-competitive randomized online
  algorithm for covering integer programs with row-sparsity $k$ and
  column-sparsity $\ell$.
\end{theorem}
Again, we note that the algorithm does not assume knowledge of the
eventual $k$ or $\ell$ values; it works with the current values after
each constraint. Furthermore, the algorithm clearly does not need the
entire cost function in advance: it suffices to know the cost
coefficient $c_i$ of each variable $i$ at the arrival time of the first
constraint that contains $i$.

\bibliographystyle{plain}
\bibliography{covering-pd}

\begin{thebibliography}{10}

\bibitem{AAABN03}
Noga Alon, Baruch Awerbuch, Yossi Azar, Niv Buchbinder, and Joseph~(Seffi)
  Naor.
\newblock The online set cover problem.
\newblock In {\em STOC}, pages 100--105, 2003.

\bibitem{AAABN-talg06}
Noga Alon, Baruch Awerbuch, Yossi Azar, Niv Buchbinder, and Joseph~(Seffi)
  Naor.
\newblock A general approach to online network optimization problems.
\newblock {\em ACM Trans. Algorithms}, 2(4):640--660, 2006.

\bibitem{AS-book}
Noga Alon and Joel Spencer.
\newblock {\em The Probabilistic Method}.
\newblock Wiley-Interscience, New York, 2008.

\bibitem{BBN-focs07-paging}
Nikhil Bansal, Niv Buchbinder, and Joseph Naor.
\newblock A primal-dual randomized algorithm for weighted paging.
\newblock In {\em FOCS}, pages 507--517, 2007.

\bibitem{BBN-stoc08-caching}
Nikhil Bansal, Niv Buchbinder, and Joseph~(Seffi) Naor.
\newblock Randomized competitive algorithms for generalized caching.
\newblock In {\em S{TOC}'08}, pages 235--244. ACM, New York, 2008.

\bibitem{BKNS10}
Nikhil Bansal, Nitish Korula, Viswanath Nagarajan, and Aravind Srinivasan.
\newblock On {\it k}-column sparse packing programs.
\newblock In {\em IPCO}, pages 369--382, 2010.

\bibitem{BN-mono}
Niv Buchbinder and Joseph~(Seffi) Naor.
\newblock The design of competitive online algorithms via a primal-dual
  approach.
\newblock {\em Found. Trends Theor. Comput. Sci.}, 3(2-3):93--263, 2007.

\bibitem{BN-MOR}
Niv Buchbinder and Joseph~(Seffi) Naor.
\newblock Online primal-dual algorithms for covering and packing.
\newblock {\em Math. Oper. Res.}, 34(2):270--286, 2009.

\bibitem{CFLP}
Robert~D. Carr, Lisa~K. Fleischer, Vitus~J. Leung, and Cynthia~A. Phillips.
\newblock Strengthening integrality gaps for capacitated network design and
  covering problems.
\newblock In {\em SODA}, pages 106--115, 2000.

\bibitem{CGK10}
Deeparnab Chakrabarty, Elyot Grant, and Jochen K{\"o}nemann.
\newblock On column-restricted and priority covering integer programs.
\newblock In {\em IPCO}, pages 355--368, 2010.

\bibitem{KY05}
Stavros~G. Kolliopoulos and Neal~E. Young.
\newblock Approximation algorithms for covering/packing integer programs.
\newblock {\em J. Comput. Syst. Sci.}, 71(4):495--505, 2005.

\bibitem{Korman05}
Simon Korman.
\newblock On the use of randomness in the online set cover problem.
\newblock {\em M.Sc. thesis, Weizmann Institute of Science}, 2005.

\bibitem{KY09}
Christos Koufogiannakis and Neal~E. Young.
\newblock Greedy $\delta$-approximation algorithm for covering with arbitrary
  constraints and submodular cost.
\newblock In {\em ICALP (1)}, pages 634--652, 2009.

\bibitem{PC11}
David Pritchard and Deeparnab Chakrabarty.
\newblock Approximability of sparse integer programs.
\newblock {\em Algorithmica}, 61(1):75--93, 2011.

\bibitem{Srinivasan99}
Aravind Srinivasan.
\newblock Improved approximation guarantees for packing and covering integer
  programs.
\newblock {\em SIAM J. Comput.}, 29(2):648--670, 1999.

\bibitem{Srinivasan01}
Aravind Srinivasan.
\newblock New approaches to covering and packing problems.
\newblock In {\em SODA}, pages 567--576, 2001.

\bibitem{Srinivasan06}
Aravind Srinivasan.
\newblock An extension of the lov{\'a}sz local lemma, and its applications to
  integer programming.
\newblock {\em SIAM J. Comput.}, 36(3):609--634, 2006.

\bibitem{Y94}
Neal~E. Young.
\newblock The k-server dual and loose competitiveness for paging.
\newblock {\em Algorithmica}, 11(6):525--541, 1994.

\end{thebibliography}

\appendix

\section{Limitations of the Guess-and-Double Approach}
\label{sec:monotone-duals}

We observe here that previously used primal-dual updates (to the best of our knowledge) are insufficient to prove a
competitive ratio that depends only on $k$. A large number of online algorithms are based on monotone primal-dual
updates. Buchbinder and Naor~\cite[Lemma~3.1]{BN-MOR} showed that if we maintain monotone duals then the primal-dual
gap may be as large as $\Omega(\log \frac{a_{max}}{a_{min}})$. In order to get around this issue and obtain an $O(\log
n)$ competitive ratio for general covering LPs, \cite[Theorem~4.1]{BN-MOR} used a guess-and-double framework which uses
duals in a partly non-monotone manner. However, as we show below, this scheme does not suffice to obtain a primal-dual
gap independent of $n$, even when $k=1$.

The guess-and-double scheme proceeds in phases, and within each phase it maintains monotone primal as well as duals.
But when the phase changes, the scheme resets all dual values to zero and starts afresh; this is the only allowed dual
reduction. To maintain an approximately feasible dual, this scheme is allowed to change phases (and reset duals) only
when the primal cost increases by (say) a factor of two. Upon arrival of the first constraint $(\sum_{i\in T_1}
a_{i1}x_i\ge 1)$, the scheme produces a lower bound $\alpha_1=\min_{i\in T_1} c_i/a_{i1}$ on the optimal value and
begins its first phase. In the $r^{th}$ phase it is assumed that $\alpha_r$ is the optimal value until the primal cost
exceeds $\alpha_r$; at this point the scheme sets $\alpha_{r+1}=2\cdot \alpha_r$ and enters phase $r+1$. A competitive
ratio of $O(\beta)$ is proven via this scheme by showing that after each phase $r$, the total primal cost is at most
$\beta$ times the total dual value (added over all phases up to $r$).

\begin{lemma}
Any online algorithm using the guess-and-double framework for covering LPs (even with $k=1$) incurs an unbounded primal
to dual ratio.
\end{lemma}
\begin{proof}
It suffices to show that for every $\rho>2$, there exist instances of the online covering LP with $k=1$ where any
algorithm using the guess-and-double framework incurs a primal to dual ratio of at least $\Omega(\rho)$. Our instances
will have all costs being one, so the primal objective is just $\sum_{i=1}^n x_i$. Since $k=1$, all constraints will be
of the form $x_i\ge b$ for some $i\in[n]$ and $b>0$. The first constraint is $x_1\ge \rho^{\rho+2}$. So
$\alpha_1=\rho^{\rho+2}$ in the guess-and-double scheme. In each phase $r$, constraints appear for a completely new set
of variables $x_{r,1},x_{r,2},\ldots$ as follows. Initialize $j\gets 1$.
\begin{quote} {\bf Sequence $I(r,j)$ :} Constraints of the form $x_{r,j}\ge \rho^h$ with dual variable
$y_{r,j}(h)$ appear for $h=0,1,\ldots$, until the first time that algorithm sets dual value $y_{r,j}(h)<\rho^{h-1}$.
\end{quote}

At this point we move on to the next variable $x_{r,j+1}$, i.e. set $j\gets j+1$ and repeat the sequence $I(r,j)$.
Also, the entire phase $r$ ends when the sum of variables in this phase exceeds $\alpha_r$, at which point we abort the
current sequence $I(r,j)$ and enter phase $r+1$.

Suppose $q$ variables are used in phase $r$. Let $h_1,\ldots,h_q$ denote the number of constraints produced in
$I(r,1),\ldots,I(r,q)$ respectively. Note that the dual variables in this phase are $\bigcup_{j\in [q]} \{y_{r,j}(h) :
1\le h\le h_j\}$, dual constraints are $\sum_{h} y_{r,j}(h)/\rho^h \le 1$ for all $j\in [q]$, and dual objective is
$\sum_{j\in [q]}\sum_h y_{r,j}(h)$.
\begin{cl}\label{cl:guessD-1}
For all $j\in[q]$, $h_j\le \rho+1$.
\end{cl}
\begin{proof}
Fix any $j\in [q]$; the dual constraint corresponding to variable $x_{r,j}$ reads $\sum_{h} y_{r,j}(h)/\rho^h \le 1$.
By definition of the sequence $I(r,j)$, for all $1\le h<h_j$ the dual value $y_{r,j}(h) \ge \rho^{h-1}$. Note that
duals in a single phase are monotone-- so at the end of sequence $I(r,j)$ we have:
$$1 \quad \ge \quad \sum_{h} y_{r,j}(h)/\rho^h \quad \ge \quad \sum_{h=1}^{h_j-1} \rho^{h-1}/\rho^h \quad = \quad \frac{h_j-1}{\rho}$$
The first inequality is the dual constraint for $x_{r,j}$ and the second uses the dual values.
\end{proof}
From this claim it follows that the primal increase of each $x_{r,j}$ is at most $\rho^{\rho+1} \le \alpha_1/\rho\le
\alpha_r/\rho$. This implies that $q\ge 2$ variables are used in this phase. Note that the primal increase in phase $r$
is:
\begin{equation}\label{eq:guessD-primal}
P(r) \quad = \quad \sum_{j\in[q]} \rho^{h_j} \quad \ge \quad \alpha_r
\end{equation}
The next claim shows that the dual increase can only be a small fraction of the primal.
\begin{cl}\label{cl:guessD-2}
The total dual increase in phase $r$ is at most $\frac4{\rho}\cdot P(r)$.
\end{cl}
\begin{proof}
Consider any primal variable $x_{r,j}$, and its dual constraint \\
$\sum_{h=1}^{h_j} y_{r,j}(h)/\rho^h \le 1$. Clearly the maximum dual value achievable from these dual variables
$\sum_{h=1}^{h_j} y_{r,j}(h)\le \rho^{h_j}$.

Now consider $j\le q-1$; the sequence $I(r,j)$ was ended due to $y_{r,j}(h_j)<\rho^{h_j-1}$. Also by the dual
constraint, $y_{r,j}(h)\le \rho^h$ for all $1\le h\le h_j-1$. Thus:
$$\sum_{h=1}^{h_j} y_{r,j}(h) \quad \le \quad \rho^{h_j-1} + \sum_{h=1}^{h_j-1} \rho^h \quad \le \quad \rho^{h_j-1} \cdot
\left(1+\frac1{1-1/\rho}\right) \quad \le \quad 3\cdot \rho^{h_j-1},$$ where the last inequality uses $\rho\ge 2$. We
now obtain that the total dual value in phase $r$:
\begin{eqnarray*}
\sum_{j=1}^q \sum_{h=1}^{h_j} y_{r,j}(h) & \le & \rho^{h_q} + \sum_{j=1}^{q-1} \sum_{h=1}^{h_j}
y_{r,j}(h) \quad  \le \quad  \rho^{h_q} + \sum_{j=1}^{q-1} 3\cdot \rho^{h_j-1}\\
&\le_{\eqref{eq:guessD-primal}} & \rho^{h_q} +\frac3\rho \cdot P(r) \quad \le_{\mbox{{\small
\lref[Claim]{cl:guessD-1}}}} \quad \rho^{\rho+1} +\frac3\rho \cdot P(r) \\
& \le & \frac{\alpha_r}\rho +\frac3\rho \cdot P(r) \quad  \le_{\eqref{eq:guessD-primal}} \quad \frac4\rho \cdot P(r).
\end{eqnarray*}
This proves the claim.\end{proof}

Using \lref[Claim]{cl:guessD-2} and~\eqref{eq:guessD-primal} it follows that for the input sequence constructed above,
the total dual value accrued $\sum_r \left(\sum_{j,h} y_{r,j}(h)\right)$ is at most $4/\rho$ times the primal cost
$\sum_r P(r)$.
\end{proof}
This lemma shows that using just the dual reductions allowed within a guess-and-double framework is insufficient to
prove a primal-dual ratio independent of $n$. Instead our online algorithm  performs more sophisticated dual reduction
that is used to prove $O(\log k)$-competitiveness.


\end{document}